\documentclass[letterpaper, 10 pt, conference]{ieeeconf}

\IEEEoverridecommandlockouts
\usepackage{cite}
\usepackage{graphics}
\usepackage{epsfig}
\usepackage{ntheorem,amsmath}
\usepackage{amssymb}
\usepackage{mathtools}
\usepackage{mathdots}
\usepackage{commath}
\usepackage{graphicx,bm}
\usepackage{picture}
\usepackage{xcolor}
\usepackage{float}
\usepackage{verbatimbox}
\usepackage{booktabs}
\usepackage{listings}
\usepackage{tcolorbox}
\usepackage{url}
\usepackage{algorithm}
\usepackage{algpseudocode}
\usepackage{textcomp}
\usepackage{grffile}
\usepackage{times}
\usepackage{hyperref}
\usepackage{stfloats}

\makeatletter
\renewtheoremstyle{plain}
{\item{\theorem@headerfont ##1\ ##2\theorem@separator}~}
{\item{\theorem@headerfont ##1\ ##2\ (##3)\theorem@separator}~}
\theoremseparator{.}
\makeatother

\newtheorem{theorem}{Theorem}
\newtheorem{proposition}{Proposition}
\newtheorem{assumption}{Assumption}
\newtheorem{corollary}{Corollary}
\newtheorem{lemma}{Lemma}
\newtheorem{remark}{Remark}

\DeclareMathOperator{\col}{col}

\newcommand{\N}{{\mathbb{N}}}
\newcommand{\R}{{\mathbb{R}}}
\newcommand{\C}{{\mathbb{C}}}

\renewcommand{\QED}{\QEDopen}

\title{\LARGE \bf Data-Driven Stabilization of Continuous-Time LTI Systems\\
from Noisy Input–Output Data
}

\author{Alessandro Bosso$^{1}$, Marco Borghesi$^{1}$, Andrea Iannelli$^{2}$, Bowen Yi$^{3}$, and Giuseppe Notarstefano$^{1}$
\thanks{$^{1}$A. Bosso, M. Borghesi, and G. Notarstefano are with the Department of Electrical, Electronic, and Information Engineering, University of Bologna, Italy. Email: {\tt\small \{alessandro.bosso, m.borghesi, giuseppe.notarstefano\}@unibo.it}}
\thanks{$^{2}$A. Iannelli is with the Institute for System Theory and Automatic Control, University of Stuttgart, Germany. Email: {\tt\small andrea.iannelli@ist.uni-stuttgart.de}
}
\thanks{$^{3}$B. Yi is with the Department of Electrical Engineering, Polytechnique Montr\'eal and GERAD, Montr\'eal, Canada.
{\tt\small bowen.yi@polymtl.ca}}
\thanks{The research leading to these results has received funding from the European Union's Horizon Europe research and innovation program under the Marie Sk{\l}odowska-Curie Grant Agreement No. 101104404 - \mbox{IMPACT4Mech}.}
}

\begin{document}

    \maketitle
	\thispagestyle{empty}
	\pagestyle{empty}
	
	\begin{abstract}
        We present an approach to compute stabilizing controllers for continuous-time linear time-invariant systems directly from an input–output trajectory affected by process and measurement noise.
        The proposed output-feedback design combines (i) an observer of a non-minimal realization of the plant and (ii) a feedback law obtained from a linear matrix inequality (LMI) that depends solely on the available data.
        Under a suitable interval excitation condition and knowledge of a noise energy bound, the feasibility of the LMI is shown to be necessary and sufficient for stabilizing all non-minimal realizations consistent with the data.
        We further provide a condition for the feasibility of the LMI related to the signal-to-noise ratio, guidelines to compute the noise energy bound, and numerical simulations that illustrate the effectiveness of the approach.
	\end{abstract}
    
    \section{Introduction}

    Recent trends in automatic control highlight a growing reliance on data-based methods.
    Among them, the design of controllers directly from data---rooted in classical adaptive control \cite{ioannou2012robust}---has gained significant momentum.
    In particular, \emph{data-driven control} \cite{DataBasedCtrlBook} has become a central paradigm, inspired by the wave of results related to Willems et al.'s fundamental lemma \cite{willems2005note} and fueled by the increasing availability of large datasets.
    This approach aims to map the available data into controllers via linear matrix inequalities (LMIs) or related optimization programs, leading to an end-to-end methodology that avoids any modeling or identification step.

    In the discrete-time domain, two fundamental contributions to data-driven control are \cite{de2019formulas} and \cite{van2020data}, focusing on linear time-invariant (LTI) systems.
    These works proposed LMI formulations based respectively on data-based closed-loop system parameterizations and the data informativity framework.
    Later, data-driven LMIs particularly suited to noisy datasets have been developed from open-loop system parameterizations \cite{bisoffi2022data, van2023quadratic, DataBasedCtrlBook}.
    Notably, these approaches enable the enforcement of design specifications (e.g., closed-loop stability and $\mathcal{H}_2$/$\mathcal{H}_\infty$ performance) for all systems consistent with the data.
    Further developments include approaches for linear quadratic regulation \cite{dorfler2023certainty}, partial model knowledge \cite{berberich2022combining}, time-varying systems \cite{nortmann2023direct}, nonlinear systems \cite{monshizadeh2025versatile}, systems subject to input saturations \cite{seuret2024robust}, and output-feedback control based on the behavioral framework \cite{van2023behavioral} or non-minimal realizations constructed from shifted input--output data \cite{alsalti2025notes}.

    In the continuous-time setting, similar LMIs have been developed for data-driven state-feedback control \cite{de2019formulas, eising2024sampling, hu2025enforcing}, also accounting for the presence of noise, but at the price of requiring state derivative measurements in the dataset.
    To circumvent this limitation, research efforts have been dedicated to achieving \emph{derivative-free} state-feedback approaches \cite{rapisarda2023orthogonal, ohta2024sampling}.
    Despite these advancements, data-driven control of continuous-time systems from input--output data remains far less developed.
    In this direction, \cite{bosso2025derivative} and \cite{bosso2025data} proposed a derivative-free framework to design output-feedback controllers with a filter of the input--output data acting as an observer of a non-minimal realization of the plant.
    While this approach achieves stabilization and output regulation in the noise-free setting, to the best of the authors' knowledge, no method can yet handle noisy input--output data.

    Motivated by this gap, this paper proposes an approach for data-driven stabilization of a class of multi-input multi-output (MIMO) systems subject to process and measurement noise that generalizes the results of \cite{bosso2025data}.
    The proposed control design comprises a filter of the input--output signals and a linear feedback law depending on the filter states.
    The gains of the feedback law are computed with an LMI obtained by postprocessing the offline data through the filter dynamics.
    As compared with \cite{bosso2025data}, we leverage the quadratic matrix inequality (QMI) formalism of \cite{van2023quadratic} to obtain a robust LMI formulation whose decision variables do not scale with the data, thus also ensuring application to large datasets.
    Under an interval excitation condition and assuming the knowledge of a noise energy bound, we use a matrix version of the S-lemma to prove that the feasibility of the LMI is necessary and sufficient for stabilizing all non-minimal realizations consistent with the data, characterized by a bounded matrix ellipsoid.
    We then relate the feasibility of the LMI to the signal-to-noise ratio of the dataset.
    Finally, we show a method to compute the noise energy based on assumptions on the noise and the tuning of the filters.

    The paper is organized as follows.
    In Section \ref{sec:problem}, we introduce the considered control problem.
    In Section \ref{sec:main_result}, we illustrate our approach and its theoretical guarantees.
    In Section \ref{sec:feasibility}, we provide a sufficient condition for LMI feasibility, while Section \ref{sec:noise_bounds} shows how to compute the noise energy bound.
    Finally, Section \ref{sec:simulations} presents two numerical examples and Section \ref{sec:conclusion} concludes the paper.
    Further details and the technical proofs are left in the Appendix.
        
    \subsubsection*{Notation}
    We use $\N$, $\R$, and $\C$ to denote the sets of natural, real, and complex numbers, respectively, while $\mathrm{i}$ is the imaginary unit.
    $I_j$ denotes the identity matrix of dimension $j$, while $0_j$ and $0_{j \times k}$ denote zero matrices of dimensions $j \times j$ and $j \times k$.
    $\otimes$ is the Kronecker product between matrices.
    We denote by $\sigma(R)$ the spectrum of a square matrix $R$ and by $\lambda_{\min}(S)$ and $\lambda_{\max}(S)$ the minimum and maximum eigenvalues of a symmetric matrix $S$.
    Given a symmetric matrix $S$, $S\succ 0$ (resp., $S \succeq 0$) denotes that it is positive definite (resp., positive semidefinite).
    Similarly, $\prec 0$ and $\preceq 0$ are used for negative definite and semidefinite matrices. 
    For symmetric matrices $S_1$ and $S_2$, we write $S_1 \preceq S_2$ if $S_1 - S_2 \preceq 0$.
    Given a signal $f(\cdot):[0, T] \to \R^{l}$, its finite-horizon $\mathcal{L}_2 [0, T]$ norm is $\| f(\cdot) \|_{2, [0, T]} \coloneqq \big(\int_{0}^{T}|f(\tau)|^2\mathrm{d}\tau\big)^{1/2}$.
    Finally, the $\mathcal{H}_{\infty}$ norm of a transfer function matrix $\mathcal{G}(s)$ is denoted by $\| \mathcal{G} \|_{\infty}$.

    \section{Problem Statement}\label{sec:problem}
    We consider continuous-time linear time-invariant systems described by the following input--output model:
    \begin{equation}\label{eq:IO_model}
        \mathcal{D}\!\left(\!\frac{\mathrm{d}}{\mathrm{d}t} \!\right)(y - v) = \mathcal{N}\!\left(\!\frac{\mathrm{d}}{\mathrm{d}t} \!\right)u + \mathcal{N}_w\!\left(\!\frac{\mathrm{d}}{\mathrm{d}t} \!\right)w,
    \end{equation}
    where $u \in \R^m$ is the control input, $y \in \R^p$ is the measured output, $w \in \R^q$ is the process noise, and $v \in \R^p$ is the measurement noise, while, for some known positive integer $n$ that we denote as \emph{order} of the differential equation \eqref{eq:IO_model},
    \begin{equation}\label{eq:polynomials}
        \begin{split}
            \mathcal{D}(\xi) &\coloneqq I_p \xi^n + A_{n - 1}\xi^{n - 1} + \ldots + A_1 \xi + A_0\\
            \mathcal{N}(\xi) &\coloneqq B_{n - 1}\xi^{n - 1} + \ldots + B_1 \xi + B_0\\
            \mathcal{N}_w(\xi) &\coloneqq E_{n - 1}\xi^{n - 1} + \ldots + E_1 \xi + E_0
        \end{split}
    \end{equation}
    are polynomial matrices with unknown matrix coefficients $A_{i} \in \R^{p \times p}$, $B_{i} \in \R^{p \times m}$, $i \in \{0, \ldots, n - 1\}$ and known matrix coefficients $E_i \in \R^{p \times q}$, $i \in \{0, \ldots, n - 1\}$.
    Note that the solutions of the differential equation \eqref{eq:IO_model} are intended in the weak sense as per \cite[Def. 2.3.7]{polderman1998introduction}, so that the signals are not required to be differentiable but only locally integrable.

    We represent the input--output model \eqref{eq:IO_model} with the following state-space realization, having state $x \in \R^{np}$:
    \begin{equation}\label{eq:plant}
        \begin{split}
            \dot{x} &= \! \underbrace{\left[\!\!\!\!\begin{array}{c|c}
                \begin{array}{ccc}
                     \!0_p\! & \cdots & \!0_p\!\!\\
                     \!I_p\! & & \\
                     & \ddots & \\
                     &  & \!I_p\!\!
                \end{array} & \begin{array}{c}
                     \!-A_0\!\\
                     \!-A_1\!\\
                     \vdots\\
                     \!\!\!\! -A_{n - 1} \!\!\!\!
                \end{array}
            \end{array}\!\!\!\right]}_{A} \!\! x + \!\! \underbrace{\begin{bmatrix}
                B_0\\ B_1\\ \vdots\\ \! B_{n - 1} \!
            \end{bmatrix}}_{B}\!\! u + \!\! \underbrace{\begin{bmatrix}
                E_0 \\ E_1 \\ \vdots \\ \! E_{n - 1} \!
            \end{bmatrix}}_{E} \! w\\
            y &= \underbrace{\left[\begin{array}{ccc|c}
                0_p & \cdots & 0_p & I_p
            \end{array}\right]}_{C} x + v,
        \end{split}
    \end{equation}
    where the pair $(C, A)$ is observable.
    Further details on the realization \eqref{eq:plant} are provided in Appendix \ref{sec:modeling}.
    We also assume the following property.
    \begin{assumption}\label{hyp:ctrb}
        The pair $(A, B)$ is controllable.
    \end{assumption}
    
    Suppose that an experiment is performed on system \eqref{eq:plant} over the time interval $[0, T]$, with $T > 0$.
    Specifically, given an unknown initial condition $x(0) \in \R^{np}$, an input trajectory $u(\cdot): [0, T] \to \R^m$, and unknown noise signals $w(\cdot): [0, T] \to \R^{q}$, $v(\cdot): [0, T] \to \R^p$, let the corresponding state and output trajectories be $x(\cdot): [0, T] \to \R^{np}$ and $y(\cdot): [0, T] \to \R^p$, with $x(\cdot)$ an absolutely continuous function.
    Then, we suppose that the following continuous-time dataset is collected during the experiment:
    \begin{equation}\label{eq:dataset}
        \{(t, u(t), y(t)): t \in [0, T]\}.
    \end{equation}
    Our goal is to use \eqref{eq:dataset} to design a stabilizing controller for system \eqref{eq:IO_model} (equivalently, system \eqref{eq:plant}) without any prior knowledge of the matrices $A_i$, $B_i$, $i \in \{1, \ldots, n -1 \}$ except for the order $n$.
    Specifically, we aim to compute the matrices $A_{\textup{c}}$, $B_{\textup{c}}$, $C_{\textup{c}}$, and $D_{\textup{c}}$ of a dynamic controller of the form
    \begin{equation}\label{eq:ctrl_generic}
        \begin{split}
            \dot{x}_{\textup{c}} &= A_{\textup{c}}x_{\textup{c}} + B_{\textup{c}}y\\
            u&= C_{\textup{c}}x_{\textup{c}} + D_{\textup{c}}y,
        \end{split}
    \end{equation}
    with $x_{\textup{c}} \in \R^\mu$ the controller state, so that the state matrix of the closed-loop interconnection of \eqref{eq:plant} and \eqref{eq:ctrl_generic}
    \begin{equation}\label{eq:closed-loop_matrix}
        \begin{bmatrix}
            A + BD_{\textup{c}}C & BC_{\textup{c}}\\
            B_{\textup{c}}C & A_{\textup{c}}
        \end{bmatrix}
    \end{equation}
    is Hurwitz.

    \section{Data-Driven Stabilization from Noisy Input--Output Data}\label{sec:main_result}
    
    \begin{algorithm}[!b]
    \caption{Data-Driven Stabilization from Noisy I--O Data}\label{alg:stabilization}
        \begin{algorithmic}
        	\State \hspace{-0.47cm} \textbf{Initialization} 
        	    \State \emph{Measured dataset:}
                \begin{equation}
                    \{(t, u(t), y(t)): t \in [0, T]\}.
                \end{equation}
                \State \emph{Controller state dimension:} $\mu \coloneqq n(p + m)$.
                \State \emph{Tuning:} $\Lambda \in \R^{n \times n}$ Hurwitz with $n$ distinct eigenvalues; $\Gamma \in \R^n$ such that $(\Lambda, \Gamma)$ is controllable; $\Delta \in \R^{p \times p}$, $\Delta = \Delta^\top \succeq 0$.
                Filter matrices:
                \begin{equation}\label{eq:FGL}
                    F \coloneqq I_{p + m} \otimes \Lambda, \quad G \coloneqq \! \begin{bmatrix}
                        0_{np \times m}\\
                        I_m \otimes \Gamma
                    \end{bmatrix}\!, \quad L \coloneqq \! \begin{bmatrix}
                        I_p \otimes \Gamma\\
                        0_{nm \times p}
                    \end{bmatrix}\!.
                \end{equation}
        	\State \hspace{-0.47cm} \textbf{Filtering}
                \State \emph{Filter signals:} compute for all $t \in [0, T]$:
                \begin{equation}\label{eq:z_chi}
                    \zeta(t) \coloneqq \begin{bmatrix}
                        \chi(t)\\
                        \hat{z}(t)
                    \end{bmatrix} \coloneqq \begin{bmatrix}
                        e^{\Lambda t}\Gamma\\
                        \int_{0}^{t}e^{F(t - \tau)}(Gu(\tau) + Ly(\tau))\mathrm{d}\tau
                    \end{bmatrix}\!.
                \end{equation}
            \State \hspace{-0.47cm} \textbf{Stabilizing Gain Computation}
                \State \emph{LMI:}
                find $P \in \R^{\mu \times \mu}$ and $Q \in \R^{m \times \mu}$ satisfying the linear matrix inequality \eqref{eq:LMI} reported in the next page.
                \State \emph{Control gain:}
                \begin{equation}\label{eq:K}
                    K = QP^{-1}.
                \end{equation}
                \State \hspace{-0.47cm} \textbf{Control Deployment}
                \State \emph{Control law:}
                \begin{equation}\label{eq:controller}
                    \dot{x}_{\textup{c}} = (F + GK)x_{\textup{c}} + Ly, \qquad  u = Kx_{\textup{c}},
                \end{equation}
                with arbitrary $x_{\textup{c}}(0) \in \R^{\mu}$.
        \end{algorithmic}
    \end{algorithm}
    \begin{figure*}[!b]
    \hrulefill
        \begin{equation}\tag{$\mathcal{S}$}\label{eq:LMI}
            \int_{0}^{T}\!\left[\begin{array}{c}
            \! Ly(\tau) \!\\ \hline
            \! -\zeta(\tau) \!
        \end{array}\right]\!\!\left[\begin{array}{c}
            \! Ly(\tau) \!\\ \hline
            \! -\zeta(\tau) \!
        \end{array}\right]^\top \!\!\!\! \mathrm{d}\tau -\left[\begin{array}{c|c}
                \!\!L\Delta L^\top\!\!  + FP + PF^\top\! + GQ + Q^\top G^\top\!\!&  \begin{array}{cc}
                    \!\!\!0_{\mu \times n}\! & \!P\!\!\!\!
                \end{array} \\ \hline
                \begin{array}{c}
                    0_{n \times \mu}\\
                    P
                \end{array}& 0_{n + \mu}
            \end{array} \right]\! \succ 0, \qquad P = P^\top\! \succ 0.
        \end{equation}
    \end{figure*}

    We now illustrate in detail the proposed data-driven stabilization approach, whose complete implementation is provided in Algorithm \ref{alg:stabilization}.
    
    The central idea is to build a non-minimal realization of the plant (Section \ref{sec:realization}) that has a structure amenable to observer design.
    This allows us to split the dynamic controller into an observer of the non-minimal realization and a feedback law depending on the observer states (Section \ref{sec:observer}).
    In particular, the feedback gain is designed to quadratically stabilize all realizations consistent with the data (Section \ref{sec:E}) by employing a data-based LMI derived from the matrix S-lemma \cite[Thm. 4.10]{van2023quadratic} (Section \ref{sec:LMI}).

    \subsection{Non-Minimal Realization for Data-Driven Control}\label{sec:realization}
    Inspired by classical results in adaptive control \cite[Ch. 4]{narendra1989stable} and following the state-space framework developed in \cite{bosso2025data}, we now introduce a non-minimal realization of the transfer function matrix $\mathcal{D}^{-1}(s)\mathcal{N}(s)$, characterizing the noise-free input--output behavior of \eqref{eq:IO_model}.
    Let $(\Lambda, \Gamma)$ be a controllable pair as described in Algorithm \ref{alg:stabilization}, and let
    \begin{equation}\label{eq:D_lambda}
        \mathcal{D}_{\lambda}(\xi) \coloneqq \xi^n + \lambda_{n - 1}\xi^{n - 1} + \ldots + \lambda_1 \xi + \lambda_0
    \end{equation}
    be the characteristic polynomial of $\Lambda$, having $n$ distinct roots with negative real part.

    Then, define $F \in \R^{\mu \times \mu}$, $G \in \R^{\mu \times m}$, $L \in \R^{\mu \times p}$ as in \eqref{eq:FGL}, with $\mu \coloneqq n(p + m)$ that represents, as shown below, the dimension of the non-minimal realization.
    The following result will be used throughout the paper.
    The proof is based on \cite{bosso2025data} and can be found in Appendix \ref{sec:realization_eq}.
    \begin{lemma}\label{lem:realization_eq}
        Under Assumption \ref{hyp:ctrb} and given $F$, $G$, and $L$ as in \eqref{eq:FGL}, there exist full-rank matrices $\Pi \in \R^{np \times \mu}$, $H \in \R^{p \times \mu}$ that satisfy
        \begin{equation}\label{eq:realization_eq}
            \begin{split}
                \Pi(F + LH) &= A\Pi, \quad \Pi G = B\\
                H &= C\Pi,
            \end{split}
        \end{equation}
        and, in addition, ensure the following properties:
        \begin{enumerate}
            \item $(F + LH, G)$ is controllable.
            \item $A - \Pi L C$ has the form
            \begin{equation}\label{eq:tilde_Lambda}
                \tilde{\Lambda} \coloneqq A - \Pi LC = \! \left[\!\!\!\!\begin{array}{c|c}
                \begin{array}{ccc}
                     0_p & \cdots & 0_p\\
                     I_p & & \\
                     & \ddots & \\
                     &  & I_p
                \end{array} \!\!\! & \! \begin{array}{c}
                     -\lambda_0 I_p\\
                     -\lambda_1 I_p\\
                     \vdots\\
                     \!\! -\lambda_{n - 1}I_p \!\!
                \end{array}
            \end{array}\!\!\!\right]\!.
            \end{equation}
        \end{enumerate}
    \end{lemma}  
    Using $\Pi$ and $H$ from Lemma \ref{lem:realization_eq}, we obtain from \cite[Lemma 1]{bosso2025data} that the following dynamical system
    \begin{equation}\label{eq:canonical}
        \begin{split}
            \dot{z} &= (F + LH)z + Gu\\
            y &= Hz,
        \end{split}
    \end{equation}
    with $z \in \R^{\mu}$, is a controllable and detectable (but not observable) non-minimal realization of $\mathcal{D}^{-1}(s)\mathcal{N}(s)$.
    In particular, the controllable and observable subsystem of \eqref{eq:canonical}, parameterized via $x = \Pi z$, obeys $\dot{x} = Ax + Bu$, $y = Cx$.

    \subsection{Observer-Based Control Design}\label{sec:observer}
    The special structure of \eqref{eq:canonical}, where $F$ is Hurwitz by construction (see Algorithm \ref{alg:stabilization}) and $Hz$ can be replaced by $y$ in the differential equation, allows us to design an observer of the plant \eqref{eq:plant} in the coordinates of the non-minimal realization, without any knowledge of the parameters in $A$ and $B$.
    In particular, define the following dynamical system:
    \begin{equation}\label{eq:observer}
        \dot{\hat{z}} = F\hat{z} + Gu + Ly,
    \end{equation}
    having the same structure of \eqref{eq:canonical}, with user-defined matrices $F$, $G$, and $L$ as per Algorithm \ref{alg:stabilization}.
    
    Then, define the estimation error
    \begin{equation}\label{eq:tilde_x}
        \tilde{x} \coloneqq x - \Pi \hat{z}.
    \end{equation}
    Using \eqref{eq:plant}, \eqref{eq:realization_eq}, \eqref{eq:observer}, and \eqref{eq:tilde_x}, the observer dynamics become
    \begin{equation}
        \begin{split}
            \dot{\hat{z}} &= F\hat{z} + Gu + L(Cx + v)\\
            &= F\hat{z} + Gu + LC(\tilde{x} + \Pi \hat{z}) + Lv\\
            &=(F + LH)\hat{z} + Gu + LC\tilde{x} + Lv,
        \end{split}
    \end{equation}
    and then we obtain the following estimation error dynamics:
    \begin{equation}\label{eq:error_dyn}
        \begin{split}
            \dot{\tilde{x}} =\;& A(\tilde{x} + \Pi \hat{z}) + Bu + E w\\
            &- \Pi ((F + LH)\hat{z} + Gu + LC\tilde{x} + Lv)\\
            =\;& (A - \Pi L C)\tilde{x} + E w - \Pi L v\\
            =\;& \tilde{\Lambda}\tilde{x} + E w - \Pi L v.
        \end{split}
    \end{equation}
    The interconnection of \eqref{eq:plant} and \eqref{eq:observer} can be rewritten in the coordinates $(\tilde{x}, \hat{z})$ as follows:
    \begin{equation}\label{eq:tilde_x_z}
        \begin{bmatrix}
            \dot{\tilde{x}}\\ \dot{\hat{z}}
        \end{bmatrix}\!\! = \!\! \begin{bmatrix}
            \tilde{\Lambda} & 0_{np \times \mu}\\
            LC\! & F\! + \! LH
        \end{bmatrix}\!\!\begin{bmatrix}
            \tilde{x} \\ \hat{z}
        \end{bmatrix} + \begin{bmatrix}
            0_{np \times m} \\ G
        \end{bmatrix}\! u \; + \begin{bmatrix}
            E & -\Pi L\\
            0_{\mu \times q}\!\!\!\! & L
        \end{bmatrix}\!\!\begin{bmatrix}
            w \\ v
        \end{bmatrix}\!.
    \end{equation}
    Since $\tilde{\Lambda}$ in \eqref{eq:tilde_Lambda} is Hurwitz and the pair $(F + LH, G)$ is controllable, the desired controller \eqref{eq:ctrl_generic} is obtained by combining the observer \eqref{eq:observer} and the feedback law
    \begin{equation}\label{eq:feedback}
        u = K\hat{z},
    \end{equation}
    where the gain $K \in \R^{m \times \mu}$ must be chosen so that $F + LH + GK$ is Hurwitz.
    Since $H$ depends on the unknown plant parameters (see \eqref{eq:realization_eq}), we cannot design $K$ using solely the prior knowledge of the pair $(F + LH, G)$.
    Therefore, in the following, we compute $K$ with data-driven techniques using the input--output trajectory \eqref{eq:dataset}.

    \subsection{Non-Minimal Realizations Consistent with the Dataset}\label{sec:E}
    Given the trajectories $u(\cdot)$ and $y(\cdot)$ over the interval $[0, T]$, we can simulate offline the observer \eqref{eq:observer}, which acts as a linear filter that post-processes the data.
    Without loss of generality and for simplicity, we consider $\hat{z}(0) = 0$, so that we obtain $\hat{z}(\cdot): [0, T] \to \R^{\mu}$ as defined in \eqref{eq:z_chi}.
    It follows that the dataset and the simulated observer trajectory $\hat{z}(\cdot)$ satisfy \eqref{eq:tilde_x_z} for almost all $t \in [0, T]$.
    Also, for all $t \in [0, T]$, the error $\tilde{x}(t) = x(t) - \Pi\hat{z}(t)$ can be written explicitly using \eqref{eq:error_dyn} and $\tilde{x}(0) = x(0)$:
    \begin{equation}\label{eq:tilde_x_t}
        \tilde{x}(t) = e^{\tilde{\Lambda} t}x(0) + \int_{0}^{t}e^{\tilde{\Lambda}(t - \tau)}(E w(\tau) - \Pi L v(\tau))\mathrm{d}\tau.
    \end{equation}
    We now exploit a technical result adapted from \cite[Lemma 4]{bosso2025data} to rewrite the free response $e^{\tilde{\Lambda}t}x(0)$ in \eqref{eq:tilde_x_t} in a convenient form.
    The proof is reported in Appendix \ref{sec:H_0}.
    \begin{lemma}\label{lem:H_0}
        For any $x(0) \in \R^{np}$ and any $(\Lambda, \Gamma)$ as defined in Algorithm \ref{alg:stabilization}, there exists $H_0 \in \R^{p \times n}$ such that, for all $t \in [0, T]$, $Ce^{\tilde{\Lambda}t}x(0) = H_0\chi(t)$, where $\chi(t)$ is given in \eqref{eq:z_chi}.
    \end{lemma}
    Combining $y(t) = C\tilde{x}(t) + H\hat{z}(t) + v(t)$ with \eqref{eq:tilde_x_t} and Lemma \ref{lem:H_0} we obtain, for all $t \in [0, T]$,
    \begin{equation}\label{eq:true_model}
            y(t) = \Theta^\star\zeta(t) + d(t),
    \end{equation}
    where $\zeta(t)$ is defined in \eqref{eq:z_chi}, the matrix
    \begin{equation}\label{eq:Theta_star}
        \Theta^\star \coloneqq \begin{bmatrix}
            H_0 & H
        \end{bmatrix} \in \R^{p \times (n + \mu)},
    \end{equation}
    denotes the parameters derived from the unknown matrices $A$ and $B$ and the initial condition $x(0)$, while, for all $t \in [0, T]$,
    \begin{equation}\label{eq:disturbance}
        d(t) \coloneqq v(t) + C\int_{0}^{t}e^{\tilde{\Lambda}(t - \tau)}(E w(\tau) - \Pi L v(\tau))\mathrm{d}\tau.
    \end{equation}
    Note that $y(t)$ and $\zeta(t)$ are available for measurement, while $\Theta^\star$ and $d(t)$ are unknown.
    In particular, to represent the lack of knowledge of $\Theta^\star$, we now describe the set of all matrices $\Theta \in \R^{p \times (n + \mu)}$ that are consistent with the dataset \eqref{eq:dataset}.
    To obtain a set with suitable regularity properties, we impose the following assumptions on the signals appearing in \eqref{eq:true_model}.
    \begin{assumption}\label{hyp:Delta}
        The matrix $\Delta$ in Algorithm \ref{alg:stabilization} is such that
        \begin{equation}\label{eq:Delta}
            \int_{0}^{T}d(\tau)d^\top(\tau)\mathrm{d}\tau \preceq \Delta.
        \end{equation}
    \end{assumption}
    \begin{assumption}\label{hyp:E}
        The following excitation condition holds:
        \begin{equation}\label{eq:E}
            Z \coloneqq \int_{0}^{T} \zeta(\tau) \zeta^\top(\tau) \mathrm{d}\tau \succ 0.
        \end{equation}
    \end{assumption}
    \begin{remark}
        Similar to the state-feedback scenario of \cite{eising2024sampling}, Assumption \ref{hyp:Delta} allows us to represent the set of parameters $\Theta$ as the solutions of a QMI and, then, exploit the results of  \cite{van2023quadratic} for data-driven stabilization.
        In Section \ref{sec:noise_bounds}, we show how to compute $\Delta$ satisfying \eqref{eq:Delta} when certain prior knowledge about the noise signals $w(\cdot)$ and $v(\cdot)$ is available.
    \end{remark}
    \begin{remark}
        Assumption \ref{hyp:E} corresponds to the full-rank conditions appearing in the data-driven literature \cite{de2019formulas, bisoffi2022data, van2023behavioral}.
        Even though such a requirement would not be necessary to introduce the non-strict LMI of \cite[Cor. 4.13]{van2023quadratic}, it provides improved theoretical guarantees and leads to the strict LMI of \cite[Thm. 4.10]{van2023quadratic} which, as argued there, may be preferred in terms of numerical reliability. 
    \end{remark}
    For compactness of notation, define
    \begin{equation}\label{eq:data_matrix}
        \left[\begin{array}{c|c}
            Y & X^\top \\\hline
            X & Z 
        \end{array}\right] \coloneqq \int_{0}^T\left[\begin{array}{c}
            \!\!\! y(\tau) \!\!\!\\ \hline
            \! -\zeta(\tau) \!
        \end{array}\right] \!\! \left[\begin{array}{c}
            \!\!\! y(\tau) \!\!\!\\ \hline
            \! -\zeta(\tau) \!
        \end{array}\right]^\top \!\!\!\! \mathrm{d}\tau \succeq 0,
    \end{equation}
    where the lines indicate the corresponding blocks, so that $Y \in \R^{p \times p}$, $X \in \R^{(n + \mu) \times p}$, and $Z \in \R^{(n + \mu)\times (n + \mu)}$.
    
    Then, replace $d(t)$ with $y(t) - \Theta \zeta(t)$ in \eqref{eq:Delta} to obtain the set of all non-minimal realizations consistent with \eqref{eq:dataset}:
    \begin{equation}\label{eq:ellipsoid}
        \boldsymbol{\mathcal{E}} \coloneqq \left\{ \Theta \in \R^{p \times (\mu + n)} : \begin{bmatrix}
            I_p & \Theta
        \end{bmatrix} N \begin{bmatrix}
            I_p \\ \Theta^\top
        \end{bmatrix} \succeq 0 \right\},
    \end{equation}
    where
    \begin{equation}\label{eq:N}
        N \coloneqq \begin{bmatrix}
            \Delta - Y & -X^\top\\
            -X & -Z
        \end{bmatrix}.
    \end{equation}

    \subsection{Data-Driven Stabilization via the Matrix S-Lemma}\label{sec:LMI}

    To solve the stabilization problem of Section \ref{sec:problem}, we need to find a gain $K$ for the feedback law \eqref{eq:feedback} that stabilizes all systems contained in the uncertainty set $\boldsymbol{\mathcal{E}}$ of \eqref{eq:ellipsoid}.
    More precisely, we want to find a single pair of matrices $P \in \R^{\mu \times \mu}$ and $K \in \R^{m \times \mu}$ such that $P = P^\top \succ 0$ and, for all $\Theta \in \boldsymbol{\mathcal{E}}$,
    \begin{equation}\label{eq:stability}
        \left(\!\! F \! + L\Theta \!\begin{bmatrix}
            0_{n \times \mu} \\ I_{\mu}
        \end{bmatrix}\!\! + GK \!\right)\!P + P\!\left(\!\! F \! + L\Theta \!\begin{bmatrix}
            0_{n \times \mu} \\ I_{\mu}
        \end{bmatrix}\!\! + GK \!\right)^\top\!\!\!\!\! \prec \! 0,
    \end{equation}
    which implies that \eqref{eq:closed-loop_matrix} is Hurwitz by choosing the controller \eqref{eq:observer}, \eqref{eq:feedback}.
    This property can be translated into an inclusion of sets described by QMIs.
    Using \eqref{eq:stability}, we obtain that the set of all non-minimal realizations stabilized by a specific choice of $P$ and $K$ is
    \begin{equation}\label{eq:C}
        \boldsymbol{\mathcal{C}}(P,\! K) \! \coloneqq \! \left\{\!\Theta \! \in \! \R^{p \times (\mu + n)}\!\! : \!\begin{bmatrix}
            I_{\mu}\\
            \!(L\Theta)^\top\!
        \end{bmatrix}^\top\!\!\!\!\!\! M(P, K)\!\! \begin{bmatrix}
            I_{\mu}\\
            \!(L\Theta)^\top\!
        \end{bmatrix} \! \!\succ \! 0  \!\right\}\!,
    \end{equation}
    where
    \begin{equation}
        M(P, K) \coloneqq -\!\left[\!\begin{array}{c|c}
                \!\!FP \! + \! PF^\top\!\! + GKP \! + \! PK^\top\! G^\top\!\!&  \begin{array}{cc}
                    \!\!\!\!0_{\mu \times n}\! & \!P\!\!\!\!\!
                \end{array} \\ \hline
                \begin{array}{c}
                    0_{n \times \mu}\\
                    P 
                \end{array}& 0_{n + \mu}
            \end{array}\! \right]\!\!.
    \end{equation}
    Then, the property that we want to impose is the following:
    \begin{equation}\label{eq:informativity}
        \boldsymbol{\mathcal{E}} \subseteq \boldsymbol{\mathcal{C}}(P, K).
    \end{equation}
    The next statement, which relates \eqref{eq:informativity} with the feasibility of the LMI \eqref{eq:LMI} used in Algorithm \ref{alg:stabilization}, is the main result of this work.
    Its proof is based on the strict matrix S-lemma found in \cite[Thm. 4.10]{van2023quadratic}.
    \begin{theorem}\label{thm:LMI}
        Consider Algorithm \ref{alg:stabilization} and let Assumptions \ref{hyp:ctrb}, \ref{hyp:Delta}, and \ref{hyp:E} hold.
        Then, the following statements are equivalent:
        \begin{enumerate}
            \item There exist $P \in \R^{\mu \times \mu}$ and $K \in \R^{m \times \mu}$, with $P = P^\top \succ 0$, such that $\boldsymbol{\mathcal{E}} \subseteq \boldsymbol{\mathcal{C}}(P, K)$.
            \item There exist $P \in \R^{\mu \times \mu}$ and $Q \in \R^{m \times \mu}$ such that the LMI \eqref{eq:LMI} is satisfied.
        \end{enumerate}
        Moreover, if \eqref{eq:LMI} is feasible, then $K = QP^{-1}$ is such that $F + LH + GK$ is Hurwitz and, thus, the controller \eqref{eq:controller} solves the data-driven stabilization problem of Section \ref{sec:problem}.
    \end{theorem}
    \begin{proof}
        Define the following sets:
        \begin{equation}\label{eq:augmented_sets}
            \begin{split}
                \boldsymbol{\mathcal{E}}_L\! &\coloneqq\! \left\{\!\Psi\! \in \! \R^{\mu \times (n + \mu)}\!:\;\, \begin{bmatrix}
                    I_{\mu}\\ \Psi^\top
                \end{bmatrix}^\top \!\!\! N_L \begin{bmatrix}
                    I_{\mu}\\ \Psi^\top
                \end{bmatrix}\; \succeq 0 \right\}\\
                \boldsymbol{\mathcal{C}}_L(P, K)\! &\coloneqq\! \left\{\!\Psi \! \in \! \R^{\mu \times (n + \mu)}\!: \!\! \begin{bmatrix}
                    I_{\mu}\\ \Psi^\top\!
                \end{bmatrix}^\top\!\!\!\!\! M(P, K)\! \begin{bmatrix}
                    I_{\mu}\\ \Psi^\top\!
                \end{bmatrix} \!\! \succ \! 0 \! \right\}\!\!,
            \end{split}
        \end{equation}
        where
        \begin{equation}
            \begin{split}
                N_L \! &\coloneqq \! \begin{bmatrix}
                L\!\!\! & \\
                & I_{n + \mu}
            \end{bmatrix}\! N \!\begin{bmatrix}
                L^\top\!\!\!\!\! & \\
                & I_{n + \mu}
            \end{bmatrix} \! = \! \begin{bmatrix}
                L(\Delta - Y)L^\top \!\! & -LX^\top\\
                -XL^\top & -Z
            \end{bmatrix}\!.
            \end{split}
        \end{equation}
        By \cite[Thm. 3.4]{van2023quadratic}, $L\boldsymbol{\mathcal{E}} \coloneqq \{L\Theta: \Theta \in \boldsymbol{\mathcal{E}}\} = \boldsymbol{\mathcal{E}}_L$ because $Z \succ 0$ by Assumption \ref{hyp:E}.
        Furthermore, $L\boldsymbol{\mathcal{C}}(P, K) \coloneqq \{L\Theta: \Theta \in \boldsymbol{\mathcal{C}}(P, K)\} \subseteq \boldsymbol{\mathcal{C}}_L(P, K)$ because, for any $\Theta \in \boldsymbol{\mathcal{C}}(P, K)$, $\Psi = L\Theta \in \boldsymbol{\mathcal{C}}_L(P, K)$. 

        We prove the equivalence of the statement using \cite[Thm. 4.10]{van2023quadratic}.
        To verify its assumptions, we inspect the matrix $N_L$:
        \begin{itemize}
            \item The $(2, 2)$ block is negative definite as $Z \succ 0$.
            \item $L(\Delta - Y)L^\top - LXZ^{-1}XL^\top \succeq 0$ because $\boldsymbol{\mathcal{E}}_L$ is nonempty, see \cite[Eq. (3.4)]{van2023quadratic}.
            \item $\ker Z \subseteq \ker LX^\top$ because $Z$ has full rank.
        \end{itemize}
        We are ready to prove both directions of the equivalence.

        \textbf{1} $\implies$ \textbf{2}: Let $P = P^\top \succ 0$ and $K$ satisfy \eqref{eq:informativity}.
        It holds that
        \begin{equation}
            \boldsymbol{\mathcal{E}}_L = L\boldsymbol{\mathcal{E}} \subseteq L\boldsymbol{\mathcal{C}}(P, K) \subseteq \boldsymbol{\mathcal{C}}_L(P, K).
        \end{equation}
        From \cite[Thm. 4.10]{van2023quadratic}, $\boldsymbol{\mathcal{E}}_L \subseteq \boldsymbol{\mathcal{C}}_L(P, K)$ holds only if there exists $\alpha \geq 0$ such that
        \begin{equation}\label{eq:LMI_alpha}
            M(P, K) - \alpha N_L \succ 0.
        \end{equation}
        Note that, by construction, $M(P, K)$ does not have full rank, thus the feasibility of \eqref{eq:LMI_alpha} implies $\alpha > 0$.
        As a consequence, divide \eqref{eq:LMI_alpha} by $\alpha$ and redefine $\alpha^{-1}P$ and $\alpha^{-1}PK$ as $P$ and $Q$ to obtain the LMI \eqref{eq:LMI}.

        \textbf{2} $\implies$ \textbf{1}: Let $P = P^\top \succ 0$ and $Q$ satisfy \eqref{eq:LMI}.
        If we define $K = QP^{-1}$, we obtain that \eqref{eq:LMI_alpha} is satisfied with $\alpha = 1$.
        Then, by \cite[Thm. 4.10]{van2023quadratic}, $\boldsymbol{\mathcal{E}}_L \subseteq \boldsymbol{\mathcal{C}}_L(P, K)$.
        From the fact that $\boldsymbol{\mathcal{E}}_L = L\boldsymbol{\mathcal{E}}$, this inclusion corresponds to saying that if $\Theta \in \boldsymbol{\mathcal{E}}$, then $L\Theta \in \boldsymbol{\mathcal{C}}_{L}(P, K)$, i.e., it satisfies \eqref{eq:stability}.
        It follows that $\Theta \in \boldsymbol{\mathcal{C}}(P, K)$, which implies \eqref{eq:informativity}.
    \end{proof}
    
    \section{Feasibility of the LMI}\label{sec:feasibility}
    We now exploit Theorem \ref{thm:LMI} to provide a sufficient condition based on data such that the LMI \eqref{eq:LMI} is feasible.

    Under Assumption \ref{hyp:E}, $\boldsymbol{\mathcal{E}}$ can be rewritten as a bounded matrix ellipsoid.
    By defining $\Tilde{\Theta} \coloneqq \Theta - \hat{\Theta}$, where
    \begin{equation}
        \hat{\Theta} \coloneqq -X^\top Z^{-1}
    \end{equation}
    is the \emph{least-squares estimate} of $\Theta^\star$, we can change the coordinates in the QMI of \eqref{eq:ellipsoid} to obtain
    \begin{equation}
        \begin{bmatrix}
            I_p & \Theta
        \end{bmatrix} N \begin{bmatrix}
            I_p \\ \Theta^\top
        \end{bmatrix} = \begin{bmatrix}
            I_p & \tilde{\Theta}
        \end{bmatrix} \!\! \begin{bmatrix}
            S_N & 0_{p \times (n + \mu)}\\
            0_{(n + \mu)\times p} & -Z
        \end{bmatrix} \!\! \begin{bmatrix}
            I_p \\ \tilde{\Theta}^\top
        \end{bmatrix}\!,
    \end{equation}
    where $S_N \coloneqq \Delta - Y + X^\top Z^{-1} X$ is the Schur complement of the $(2, 2)$ block of $N$ in \eqref{eq:N}. 
    As a consequence, \eqref{eq:ellipsoid} becomes
    \begin{equation}
        \boldsymbol{\mathcal{E}} = \left\{\Theta \in \R^{p \times (n + \mu)}: (\Theta - \hat{\Theta})Z(\Theta - \hat{\Theta})^\top \preceq S_N \right\}.
    \end{equation}
    Since the matrix \eqref{eq:data_matrix} is positive semidefinite and $Z \succ 0$, it holds from the Schur complement that $Y - X^\top Z^{-1}X \succeq 0$.
    As a consequence,
    \begin{equation}
        \boldsymbol{\mathcal{E}} \subseteq \hat{\boldsymbol{\mathcal{E}}} \coloneqq \{\Theta \in \R^{p \times (n + \mu)}: (\Theta - \hat{\Theta})Z(\Theta - \hat{\Theta})^\top \preceq \Delta \}.
    \end{equation}
    Exploiting the fact that, for any $\Theta \in \hat{\boldsymbol{\mathcal{E}}}$,
    \begin{equation}
        \lambda_{\min}(Z)(\Theta - \hat{\Theta})(\Theta - \hat{\Theta})^\top \preceq (\Theta - \hat{\Theta})Z(\Theta - \hat{\Theta})^\top \preceq \Delta,
    \end{equation}
    we obtain the following bound:
    \begin{equation}\label{eq:rho}
        |\Theta - \hat{\Theta}|^2 \leq \frac{\lambda_{\max}(\Delta)}{\lambda_{\min}(Z)} \eqqcolon \rho,
    \end{equation}
    where, for $\rho > 0$, $\textup{SNR} \coloneqq \rho^{-1}$ is the worst-case \emph{signal-to-noise ratio} according to Assumptions \ref{hyp:Delta} and \ref{hyp:E}.
    Note that $\Theta^\star \in \boldsymbol{\mathcal{E}}$, thus it follows from the triangle inequality that
    \begin{equation}\label{eq:Theta_tilde_bound}
        |\Theta - \Theta^\star| = |(\Theta - \hat{\Theta}) - (\Theta^\star - \hat{\Theta})| \leq 2\sqrt{\rho}.
    \end{equation}
    We achieve the following feasibility condition, whose proof is given in Appendix \ref{sec:SNR}.
    \begin{corollary}\label{cor:SNR}
        Let Assumption \ref{hyp:ctrb} hold.
        Then, there exists $\rho^\star > 0$ such that, if the data satisfy Assumptions \ref{hyp:Delta} and \ref{hyp:E}, with $\rho \coloneqq \lambda_{\max}(\Delta)/\lambda_{\min}(Z) \in [0, \rho^\star]$, then the LMI \eqref{eq:LMI} is feasible.
    \end{corollary}

    \section{Remarks on the Noise Bound Computation}\label{sec:noise_bounds}
    Note that the bound $\Delta$ in Assumption \ref{hyp:Delta} depends both on the noise signals $w(\cdot)$, $v(\cdot)$ and the parameters of the plant \eqref{eq:plant} and the observer \eqref{eq:observer}.
    Specifically, $d(t)$ in \eqref{eq:disturbance} can be split as $d(t) = d_w(t) + d_v(t)$, where $d_w(t)$ and $d_v(t)$ are the outputs of the following dynamical systems:
    \begin{align}
        \label{eq:d_w}
        \dot{\eta}_w &= \tilde{\Lambda}\eta_w + Ew, \qquad d_w = C\eta_w\\
        \label{eq:d_v}
        \dot{\eta}_v &= \tilde{\Lambda}\eta_v - \Pi L v, \qquad d_v = C\eta_v + v,
    \end{align}
    with zero initial conditions.
    It is convenient to report the input--output description of these systems via their transfer function matrices:
    \begin{equation}
        \mathcal{G}_{w}(s) \coloneqq \frac{\mathcal{N}_w(s)}{\mathcal{D}_{\lambda}(s)}, \qquad \mathcal{G}_{v}(s) \coloneqq \frac{\mathcal{D}(s)}{\mathcal{D}_{\lambda}(s)},
    \end{equation}
    where $\mathcal{D}_{\lambda}$ is given in \eqref{eq:D_lambda}.
    $\mathcal{G}_w(s)$ and $\mathcal{G}_v(s)$ follow from similar computations to those in Appendix \ref{sec:modeling} after noticing that, due to \eqref{eq:plant} and \eqref{eq:tilde_Lambda}, $\Pi L = [\lambda_0 I_p - A_0^\top \; \ldots \; \lambda_{n - 1}I_p - A_{n - 1}^\top]^\top$.
    Note that only $\mathcal{G}_{v}(s)$ depends on unknown parameters.

    Using the property $d(t)d(t)^\top \preceq |d(t)|^2 I_p$ and the triangle inequality in $\mathcal{L}_2 [0, T]$, we obtain
    \begin{equation}\label{eq:Delta_computation}
        \begin{split}
            \int_{0}^{T}\!\!\!d(\tau)d^\top\!(\tau)\mathrm{d}\tau &\!\preceq\! I_p \| d(\cdot) \|_{2, [0, T]}^2 \\
            &\!\preceq\! I_p ( \| d_w(\cdot) \|_{2, [0, T]} \! + \! \| d_v(\cdot) \|_{2,[0, T]})^2.
        \end{split}
    \end{equation}
    Thus, $\Delta$ in \eqref{eq:Delta} can be computed using \eqref{eq:Delta_computation} if upper bounds of $\| d_w(\cdot) \|_{2, [0, T]}$ and $\| d_v(\cdot) \|_{2,[0, T]}$ are known.
    Below, we show how to compute them under certain prior knowledge of $w(\cdot)$ and $v(\cdot)$ and suitable properties of $\mathcal{G}_w(s)$ and $\mathcal{G}_v(s)$.

    \subsection{Bound on the Filtered Process Noise}
    Suppose that a scalar $\delta_w \geq 0$ is known such that
    \begin{equation}\label{eq:delta_w}
        \| w(\cdot) \|_{2, [0, T]}^2 \leq \delta_w.
    \end{equation}
    Then, a bound on $\|d_w(\cdot) \|_{\mathcal{L}_2, [0, T]}$ is obtained by finding an upper bound $\gamma > 0$ of the $\mathcal{L}_2 [0, T]$ gain of system \eqref{eq:d_w}, which ensures:
    \begin{equation}\label{eq:d_w_bound}
        \|d_w(\cdot) \|_{2, [0, T]} \leq \gamma \|w(\cdot) \|_{2, [0, T]} \leq \gamma \sqrt{\delta_w}.
    \end{equation}
    By \cite[Thm. 3.7.4]{green1995linear}, $\gamma$ is such that the following differential Riccati equation:
    \begin{equation}\label{eq:DRE}
        \dot{W} \! = \! -\tilde{\Lambda}^\top W - W\tilde{\Lambda} - \gamma^{-2} W E E^\top W - C^\top C, \quad W(T) \! = \! 0_{np},
    \end{equation}
    has a solution over $[0, T]$.
    Since all matrices in \eqref{eq:d_w} are known, \eqref{eq:DRE} can be used to test $\gamma$.
    Let $\gamma_{\infty} > 0$ be such that $\|\mathcal{G}_w\|_{\infty} < \gamma_{\infty}$.
    Then, a conservative choice for $\gamma$ is $\gamma_{\infty}$, because by \cite[Lemma 3.7.7]{green1995linear} \eqref{eq:DRE} with $\gamma = \gamma_{\infty}$ has a solution over $[0, T]$ for all finite $T$.
    Otherwise, a tight bound can be found by searching in the interval $[0, \gamma_{\infty}]$ for $\gamma > 0$ as small as possible such that \eqref{eq:DRE} has a solution over $[0, T]$.
    
    \subsection{Bound on the Filtered Measurement Noise}
    In this scenario, we cannot use the same arguments as before because $\Pi L$ in \eqref{eq:d_v} depends on the unknown parameters in $A$.
    However, in the special case of single-output systems $(p = 1)$, we obtain the following result.
    The proof is given in Appendix \ref{sec:spectra}.
    \begin{proposition}\label{prop:spectra}
        Suppose that $p = 1$ and let all eigenvalues of $\Lambda$ be real and such that
        \begin{equation}\label{eq:spectra}
            \min_{\varsigma \in \sigma(\Lambda)} |\varsigma| \geq \max_{\varphi \in \sigma(A)}|\varphi|.
        \end{equation}
        Then, it holds that $ \|\mathcal{G}_v \|_{\infty} = 1$.
    \end{proposition}    
    Under the assumptions of Proposition \ref{prop:spectra}, if a scalar $\delta_v \geq 0$ is known such that $\|v(\cdot) \|_{2, [0, T]}^2 \leq \delta_v$, we obtain
    \begin{equation}\label{eq:delta_v}
        \| d_v(\cdot) \|_{2, [0, T]} \leq \sqrt{\delta_v}.
    \end{equation}
    Future work will investigate the problem of computing a bound for $\| d_v(\cdot) \|_{2, [0, T]}$ in the multi-output scenario ($p > 1$).

    \section{Numerical Results}\label{sec:simulations}

    We now apply Algorithm \ref{alg:stabilization} to some case studies.
    The code has been developed in MATLAB using YALMIP \cite{Lofberg2004} and MOSEK \cite{mosek} and is available at the linked repository\footnote{\url{https://github.com/IMPACT4Mech/continuous-time_stabilization_from_noisy_data}}.

    \begin{figure}[!t]
        \centering
        \includegraphics[width=0.9\linewidth]{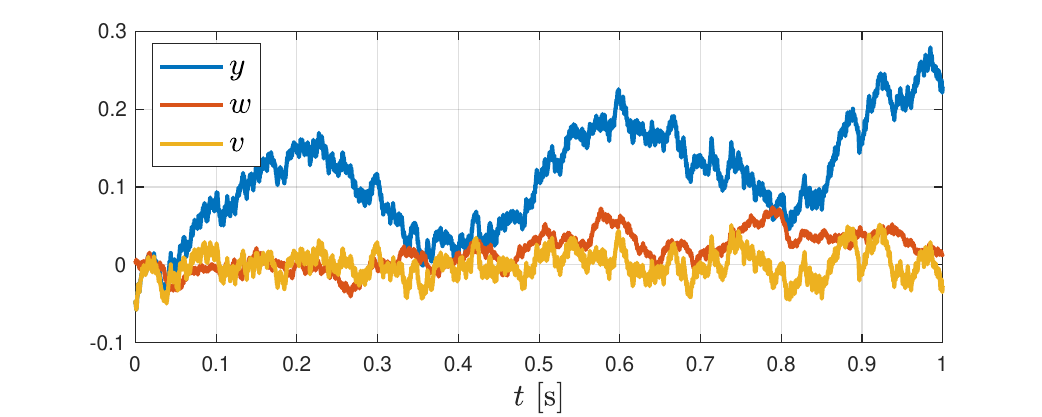}
        \includegraphics[width=0.9\linewidth]{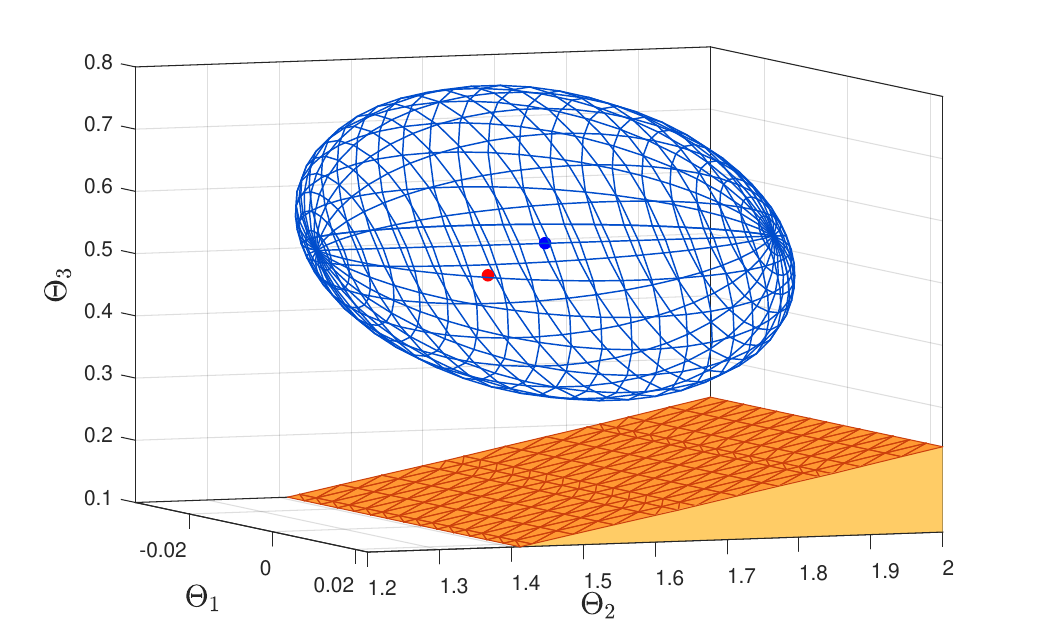}

        \vspace{-10pt}
        
        \caption{Scalar example.
        \emph{Top:} measured output $y(\cdot)$ and noise signals $w(\cdot)$ and $v(\cdot)$.
        \emph{Bottom:} representation of $\boldsymbol{\mathcal{E}}$ (blue), $\R^{3}\backslash\boldsymbol{\mathcal{C}}(P, K)$ (orange), $\Theta^\star = [0\; 1.5\; 0.5]$ (red dot), and $\hat{\Theta} = [-0.0054\; 1.6106\; 0.5395]$ (blue dot).
        }
        \label{fig:scalar_example}

        \vspace{-10pt}
        
    \end{figure}

    \subsection{Scalar System with Process and Measurement Noise}
    In the first example, we consider the following system:
    \begin{equation}\label{eq:scalar}
        \dot{x} = x + u + w, \quad y = x + v,
    \end{equation}
    which allows us to graphically illustrate the set inclusion \eqref{eq:informativity} used in Theorem \ref{thm:LMI}.
    We obtain a dataset over the interval $[0, 1]$ s by simulating \eqref{eq:scalar} from $x(0) = 0$ and applying the input $u(\cdot): t \mapsto \sin(5 \pi t)$.
    Also, we let $\| w(\cdot) \|^2_{2,[0, 1]} = \delta_w = 0.8 \times 10^{-3}$, $\| v(\cdot) \|^2_{2,[0, 1]} = \delta_v = 0.3\times 10^{-3}$.
    The noise signals and the measured output are reported in Fig. \ref{fig:scalar_example}.

    Then, we run Algorithm \ref{alg:stabilization} with $\Lambda = -2$, $\Gamma = 2$ and $\Delta = 7.1045 \times 10^{-4}$ computed according to \eqref{eq:Delta_computation}, \eqref{eq:d_w_bound}, \eqref{eq:delta_v}, with $\gamma = 0.33$ checked via \eqref{eq:DRE}.
    The resulting ellipsoid $\boldsymbol{\mathcal{E}}$ as in \eqref{eq:ellipsoid} is depicted in Fig. \ref{fig:scalar_example}, while the computed control gain $K$ and the matrix $P$ obtained from the LMI \eqref{eq:LMI} are
    \begin{equation}
        K = \begin{bmatrix*}[r]
            -29.7075\! &\!   -4.8734
        \end{bmatrix*}\!, \;\; P = \begin{bmatrix*}[r]
            1.26 & -5.34\\
            -5.34 & 53.25
        \end{bmatrix*}\! \times 10^{-3}\!,
    \end{equation}
    and the eigenvalues of \eqref{eq:closed-loop_matrix} are $\{-2, -5.37 \pm 4.34\mathrm{i}\}$.
    In Fig. \ref{fig:scalar_example}, we also depict the values of $\Theta^\star$ and $\hat{\Theta}$ and the set $\R^{3}\backslash\boldsymbol{\mathcal{C}}(P, K)$, confirming that the property \eqref{eq:informativity} holds.

    \subsection{Unstable Batch Reactor with Process Noise}
    We now consider the model of an unstable batch reactor derived from \cite[\S 2.6]{green1995linear}, which we write in the form \eqref{eq:plant} by letting  $m = p = n = 2$ and:
    \begin{equation}
        \begin{split}
            A_0 \! = \! \begin{bmatrix*}[r]
                -20.97 & -48.63\\
                2.643 &  5.867
            \end{bmatrix*}\!\!,& \; A_1 \! = \! \begin{bmatrix*}[r]
                5.297  & -10.47\\
                -0.2764 &  6.371
            \end{bmatrix*}\\
            B_0 \! = \! \begin{bmatrix*}[r]
                -59.44 & -12.63\\
                12.59 &  0.8696
            \end{bmatrix*}\!\!,& \; B_1 = \; \begin{bmatrix*}[r]
                0   &  -3.146\\
                5.679  &    0
            \end{bmatrix*}.
        \end{split}
    \end{equation}
    We assume that the process noise $w \in \R^2$ affects the system via $E_0 = I_2$, $E_1 = 0_2$, while we let $v = 0$.

    In this example, we run Algorithm \ref{alg:stabilization} several times to highlight the effect of the noise energy on the feasibility of the LMI \eqref{eq:LMI}.
    Each run considers a dataset over $[0, 3]$ s where $x(0) = 0$, $u(\cdot)$ is a sum of sinusoids (the same in each run), and $w(\cdot)$ is randomly generated.
    In particular, given five values of $\delta_w$, we generate $200$ signals uniformly extracted from the ball of radius $\sqrt{\delta_w}$ in $\mathcal{L}_2[0, 3]$, approximated via a Fourier orthonormal basis comprising the constant signal and sine/cosine functions with frequencies $2 j \pi/3$, $j \in \{1, \ldots, 100\}$.
    In all simulations, we use the filter matrices
    \begin{equation}
        \Lambda = \begin{bmatrix*}[r]
            0 & -12\\
            1 & -7
        \end{bmatrix*}, \qquad \Gamma = \begin{bmatrix*}[r]
            0 \\ 1
        \end{bmatrix*}.
    \end{equation}
    and let $\Delta = \gamma^2 \delta_w I_2$, with $\gamma = 0.07685$ checked via \eqref{eq:DRE}.
    The results are depicted in Fig. \ref{fig:rho} and show that the feasibility of \eqref{eq:LMI} decreases as $\rho$ in \eqref{eq:rho} increases.

    \begin{figure}
        \centering
        \includegraphics[width=0.99\linewidth]{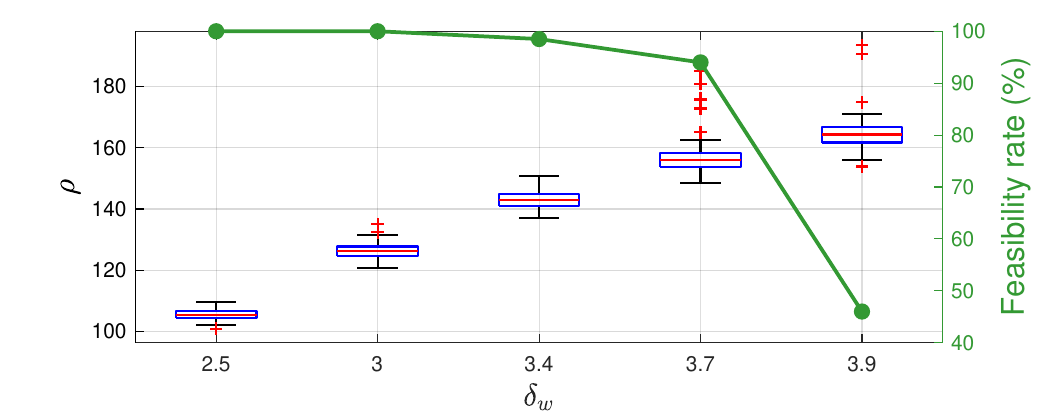}

        \vspace{-10pt}
        
        \caption{Batch reactor example. Each value of $\delta_w$ corresponds to $200$ simulation runs.
        \emph{Left axis:} box plot of $\rho$ as defined in \eqref{eq:rho}.
        \emph{Right axis:} percentage of times the LMI \eqref{eq:LMI} is feasible.}
        \label{fig:rho}

        \vspace{-10pt}
        
    \end{figure}
    
    \section{Conclusion}\label{sec:conclusion}

        We have proposed a data-driven approach to stabilize continuous-time LTI systems using solely the noisy input--output data collected from an experiment.
        The resulting algorithm involves a data-based LMI derived from the matrix S-lemma whose feasibility, under an interval excitation condition and a known noise energy bound, is both necessary and sufficient for quadratic stabilization.
        We have also related the LMI feasibility to the dataset's signal-to-noise ratio and provided tuning guidelines based on the process and measurement noise finite-horizon $\mathcal{L}_2[0, T]$ norms.
        Future work will extend the approach to broader classes of systems and consider datasets given by finite samples instead of a full continuous-time trajectory.

    \appendices\label{sec:appendix}
    \section{Remarks on the State-Space Realization}\label{sec:modeling}
    \subsection{Derivation of the Input--Output Model}
    To obtain the differential equation \eqref{eq:IO_model} from system \eqref{eq:plant}, we perform an elimination of $x$ as in \cite[\S 6.2.2]{polderman1998introduction}.
    In particular, \eqref{eq:plant} can be written as
    \begin{equation}\label{eq:state-space_ode}
        \mathcal{R}_x\left(\frac{\mathrm{d}}{\mathrm{d}t} \right)x = \mathcal{R}_y \begin{bmatrix}
            u^\top & v^\top & w^\top & y^\top
        \end{bmatrix}^\top,
    \end{equation}
    where
    \begin{equation}
        \mathcal{R}_x(\xi) \! \coloneqq \!\! \begin{bmatrix}
            \xi I_{np} - A\\
            C
        \end{bmatrix}\!\!, \; \mathcal{R}_y \! \coloneqq \! \begin{bmatrix}
            B\! & \!0_{np \times p}\! & \!E\! & \!0_{np \times p}\\
            0_{p \times m}\! & \!-I_p\! & \!0_{p \times q}\! & \!I_p
        \end{bmatrix}\!\!.
    \end{equation}
    Define the polynomial matrix
    \begin{equation}
        \mathcal{W}(\xi) \coloneqq \begin{bmatrix}
            I_p & \xi I_p & \cdots & \xi^{n-1}I_p & -\mathcal{D}(\xi)
        \end{bmatrix}.
    \end{equation}
    Then, by post-multiplying both sides of \eqref{eq:state-space_ode} by $\mathcal{W}(\mathrm{d}/\mathrm{d}t)$ and exploiting the structure of $A$, $B$, $C$, $E$ in \eqref{eq:plant}, we obtain exactly \eqref{eq:IO_model}.
    Note, in particular, that straightforward computations yield $\mathcal{W}(\xi)\mathcal{R}_x(\xi) = 0$ and $\mathcal{W}(\xi)\mathcal{R}_y = [\mathcal{N}(\xi)\; \mathcal{D}(\xi)\; \mathcal{N}_w(\xi)\; -\mathcal{D}(\xi)]$.

    \subsection{Observer Canonical Form}\label{sec:obs_form}
    Let $\mathrm{e}_i \in \R^n$, $i \in \{1, \ldots, n\}$, be the $i$-th component of the standard basis of $\R^n$ and define:
    \begin{equation}
        \bar{A} \coloneqq \begin{bmatrix}
            0_{1 \times (n - 1)} & 0\\
            I_{n - 1} & 0_{(n - 1) \times 1}
        \end{bmatrix}.
    \end{equation}
    We observe that $A$ in \eqref{eq:plant} can be written as
    \begin{equation}
        A = \bar{A} \otimes I_p + A_{\textup{m}}C = \bar{A} \otimes I_p + A_{\textup{m}}(\mathrm{e}_n^\top \otimes I_p),
    \end{equation}
    where $A_{\textup{m}} \coloneqq -[A_0^\top \; \ldots \; A_{n-1}^\top]^\top$.
    Consider the following permutation matrix:
    \begin{equation}
        \Phi \coloneqq \begin{bmatrix}
            I_p \otimes \mathrm{e}_1 & I_p \otimes \mathrm{e}_2 & \cdots & I_p \otimes \mathrm{e}_n
        \end{bmatrix} \in \R^{np \times np}.
    \end{equation}
    Some computations show that $\Phi^\top(I_p \otimes \bar{A})\Phi = \bar{A} \otimes I_p$ and $(I_p \otimes \mathrm{e}_n^\top)\Phi = \mathrm{e}_n^\top \otimes I_p = C$, which lead to the observer canonical form:
    \begin{equation}
        \begin{split}
            A_{\textup{o}} &\coloneqq \Phi A \Phi^\top = I_p \otimes \bar{A} + \tilde{A}_{\textup{m}}(I_p \otimes \mathrm{e}_n^\top)\\
            C_{\textup{o}} &\coloneqq C \Phi^\top = I_p \otimes \mathrm{e}_n^\top,
        \end{split}
    \end{equation}
    with $\tilde{A}_{\textup{m}} \coloneqq \Phi A_{\textup{m}}$.

    \section{Technical Proofs}\label{sec:proofs}
    \subsection{Proof of Lemma \ref{lem:realization_eq}}\label{sec:realization_eq}
    Since the observability indices of the pair $(C, A)$ are uniform (see \cite[As. 3]{bosso2025data}), the existence of matrices $\Pi$ and $H$ is a direct consequence of \cite[Thm. 3]{bosso2025data}.
    Then, item 1 follows from \cite[Thm. 4]{bosso2025data}.
    To prove item 2, we note from the proof of \cite[Thm. 3]{bosso2025data} that
    \begin{equation}
        \Phi(A - \Pi LC)\Phi^\top = A_{\textup{o}} - \Phi\Pi LC_{\textup{o}} = I_p \otimes (\bar{A} + \lambda_{\textup{m}}\mathrm{e}_n^\top),
    \end{equation}
    where we used the notation of Appendix \ref{sec:obs_form} and let $\lambda_{\textup{m}} \coloneqq -[\lambda_0\; \ldots \; \lambda_{n-1}]^\top$.
    We conclude the proof by noticing that
    \begin{equation}
        \begin{split}
            A - \Pi L C &= \Phi^\top (I_p \otimes (\bar{A} + \lambda_{\textup{m}}\mathrm{e}_n^\top))\Phi\\
            &= \bar{A} \otimes I_p + \Phi^\top(I_p \otimes \lambda_{\textup{m}}\mathrm{e}_n^\top)\Phi\\
            &= \bar{A} \otimes I_p + \lambda_{\textup{m}}\mathrm{e}_n^\top \otimes I_p = \tilde{\Lambda}.
        \end{split}
    \end{equation}
    \hfill\QED
    
    \subsection{Proof of Lemma \ref{lem:H_0}}\label{sec:H_0}
    
    By Lemma \ref{lem:realization_eq}, there exists a nonsingular matrix $\Phi_{\lambda}$ such that $\Phi_{\lambda} \tilde{\Lambda} \Phi_{\lambda}^{-1} = I_p \otimes \Lambda$.
    Let $\Phi_{\lambda}\tilde{x}(0) = \col(\beta_1, \ldots, \beta_p)$, with $\beta_i \in \R^n$, and let $V_i \in \R^{n \times n}$ be the solution of
    \begin{equation}
        V_i \Lambda = \Lambda V_i, \qquad V_i\Gamma = \beta_i,
    \end{equation}
    which exists and is unique by \cite[Appendix II]{bosso2025data} because $(\Lambda, \Gamma)$ is controllable.
    Define
    \begin{equation}
        V \coloneqq \begin{bmatrix}
            V_1^\top & \cdots & V_{p}^\top
        \end{bmatrix}^\top \in \R^{np \times n},
    \end{equation}
    and let $\epsilon(t) \coloneqq V\chi(t)$, where $\chi(t)$ is given in \eqref{eq:z_chi} and satisfies $\dot{\chi}(t) = \Lambda\chi(t)$, with $\chi(0) = \Gamma$. 
    We obtain
    \begin{equation}
        \dot{\epsilon}(t) \!=\! V\Lambda\chi(t) \!=\!  \Phi_{\lambda} \tilde{\Lambda} \Phi_{\lambda}^{-1}\epsilon(t), \quad \epsilon(0) \!=\! V\Gamma \!=\! \Phi_{\lambda}\tilde{x}(0),
    \end{equation}
    which yields
    \begin{equation}
        C\Phi_{\lambda}^{-1}V\chi(t) = C\Phi_{\lambda}^{-1}e^{\Phi_{\lambda}\tilde{\Lambda}\Phi_{\lambda}^{-1}}\Phi_{\lambda}\tilde{x}(0) = Ce^{\tilde{\Lambda}t}\tilde{x}(0),
    \end{equation}
    and, thus, $H_0 \coloneqq C\Phi_{\lambda}^{-1}V$.
    \hfill\QED

    \subsection{Proof of Corollary \ref{cor:SNR}}\label{sec:SNR}
    
    By Assumption \ref{hyp:ctrb}, Lemma \ref{lem:realization_eq} ensures that the pair $(F + LH, G)$ is controllable.
    As a consequence, there exist $\Omega = \Omega^\top \succ 0$ and $K$ such that
    \begin{equation}\label{eq:Lyapunov_Theta}
        \Omega\!\left(\!\! F \! + \! L\Theta^\star\!\! \begin{bmatrix}
            \! 0_{n \times \mu} \! \\ I_{\mu}
        \end{bmatrix}\!\! + \! GK \!\!\right)\! + \! \left(\!\! F \! + \! L\Theta^\star\!\! \begin{bmatrix}
            \! 0_{n \times \mu} \! \\ I_{\mu}
        \end{bmatrix}\!\! + \! GK \!\!\right)^\top\!\!\!\!\!\Omega\! \preceq \! -I_{\mu}.
    \end{equation}
    From \eqref{eq:Lyapunov_Theta}, $\Omega$ and $K$ ensure stabilization for all $\Theta$ such that
    \begin{equation}\label{eq:Theta_tilde_stable}
        \Omega L(\Theta - \Theta^\star)\begin{bmatrix}
            0_{n \times \mu} \\ I_{\mu}
        \end{bmatrix} + \begin{bmatrix}
            0_{\mu \times n} & I_{\mu}
        \end{bmatrix}(\Theta - \Theta^\star)^\top (\Omega L)^\top \preceq \ell I_\mu,
    \end{equation}
    for some $\ell \in (0, 1)$.
    Applying the norms, \eqref{eq:Theta_tilde_stable} holds if
    \begin{equation}
        |\Theta - \Theta^\star| \leq \frac{\ell}{2|\Omega L|}.
    \end{equation}
    Define
    \begin{equation}
        \rho^\star \coloneqq \frac{\ell^2}{16 |\Omega L|^2}.
    \end{equation}
    Then, if the data satisfy Assumptions \ref{hyp:Delta} and \ref{hyp:E} with $\rho \in [0, \rho^\star]$, the bound \eqref{eq:Theta_tilde_bound} implies that \eqref{eq:Theta_tilde_stable} holds, thus from the inclusions $\boldsymbol{\mathcal{E}} \subseteq \hat{\boldsymbol{\mathcal{E}}} \subseteq \boldsymbol{\mathcal{C}}(\Omega^{-1}, K)$ we obtain the statement by invoking Theorem \ref{thm:LMI}.
    \hfill\QED

    \subsection{Proof of Proposition \ref{prop:spectra}}\label{sec:spectra}

    In the following, we use $\mathfrak{R}(a)$ and $\mathfrak{I}(a)$ to denote the real and imaginary parts of a complex number $a \in \C$.
    Consider the following factorization of $\mathcal{G}_v(s)$:
    \begin{equation}
        \mathcal{G}_v(s) = \frac{\prod_{i = 1}^{n_r}(s - \vartheta_i)\prod_{j = 1}^{n_c}(s^2 - 2\mathfrak{R}(\varphi_j) s + |\varphi_j|^2)}{\prod_{k = 1}^n(s - \varsigma_k)},
        \end{equation}
    where $n_r + 2n_c = n$ and $\vartheta_i \in \R$, $\mathfrak{R}(\varphi_j) \pm \mathrm{i} \mathfrak{I}(\varphi_j) \in \C$, and $\varsigma_k \in \R$, $i \in \{1, \ldots, n_r\}$, $j \in \{1, \ldots, n_c\}$, $k \in \{1, \ldots, n\}$ denote, respectively, the real zeros, the complex conjugate zero pairs, and the poles of $\mathcal{G}_v(s)$.
    Therefore, $\mathcal{G}_v(s)$ is the product of rational functions of the form
    \begin{equation}
        \mathcal{G}_{\textup{r}}(s) \coloneqq \frac{s - \vartheta}{s - \varsigma}, \quad \mathcal{G}_{\textup{c}}(s) \coloneqq \frac{s^2 - 2\mathfrak{R}(\varphi)s + |\varphi|^2}{(s - \varsigma_1)(s - \varsigma_2)},
    \end{equation}
    where, due to condition \eqref{eq:spectra}, it holds that $|\varsigma| \geq |\vartheta|$ and $\min\{|\varsigma_1|, |\varsigma_2|\} \geq |\varphi|$.
    Thus, for all $\omega \in \R$, we obtain
    \begin{equation}
        |\mathcal{G}_{\textup{r}}(\mathrm{i}\omega)|^2 = \frac{\vartheta^2 + \omega^2}{\varsigma^2 + \omega^2} \leq 1,
    \end{equation}
    and
    \begin{equation}
        \begin{split}
            |\mathcal{G}_{\textup{c}}(\mathrm{i}\omega)|^2 &= \frac{(|\varphi|^2 - \omega^2)^2 + (2\mathfrak{R}(\varphi)\omega)^2}{(\varsigma_1^2 + \omega^2)(\varsigma_2^2 + \omega^2)}\\
            &= \frac{(|\varphi|^2 + \omega^2)^2 - (2\mathfrak{I}(\varphi)\omega)^2}{(\varsigma_1^2 + \omega^2)(\varsigma_2^2 + \omega^2)}\\
            &\leq \frac{(|\varphi|^2 + \omega^2)^2}{(\varsigma_1^2 + \omega^2)(\varsigma_2^2 + \omega^2)} \leq 1.
        \end{split}
    \end{equation}
    We conclude the proof by noticing that both the numerator and the denominator of $\mathcal{G}_v(s)$ are monic and of the same degree, thus
    \begin{equation}
        \lim_{\omega \to \infty}|\mathcal{G}_v(\mathrm{i}\omega)| = 1,
    \end{equation}
    and, as a consequence, $\| \mathcal{G}_v \|_{\infty} = \sup_{\omega \in \R}|\mathcal{G}_v(\mathrm{i}\omega)| = 1$.
    \hfill\QED

    \bibliographystyle{IEEEtran}
	
	\bibliography{data-driven_bib}

\end{document}